\documentclass[prd,article,groupedaddress]{revtex4}

\usepackage[dvips]{graphicx}
\usepackage[]{caption}
\usepackage{amsmath}
\usepackage{amssymb}
\usepackage{amsthm}
\usepackage{url}
\usepackage{dcolumn}
\usepackage{latexsym}
\usepackage{color}
\def\ket#1{{\lvert}#1\rangle}

\begin{document}
\newtheorem{df}{Definition}
\newtheorem{theo}[df]{Theorem}
\newtheorem{lem}[df]{Lemma}
\newtheorem{cor}[df]{Corollary}
\newtheorem{rema}[df]{Remark}
\newtheorem{nota}[df]{Notation}
\newtheorem{prop}[df]{Proposition}
\newtheorem{prob}[df]{Problem}

\title{The GHZ state in secret sharing and entanglement simulation}
\author{Anne Broadbent $^{1,2)}$\footnote{albroadb@iqc.ca}, Paul Robert Chouha $^{1)}$\footnote{chouhapr@iro.umontreal.ca} and Alain Tapp $^{1)}$\footnote{tappa@iro.umontreal.ca}}
\affiliation{1) D\'epartement IRO, Universit\'e de Montr\'eal, C.P.~6128, succursale centre-ville, Montr\'eal (Qu\'ebec),
Canada H3C 3J7, CANADA}
\affiliation{2) Institute for Quantum Computing,
University of Waterloo,
200 University Ave.~W., Waterloo, ON, N2L 3G1, CANADA}
\date{September 10, 2008}
\preprint{hep-th/yymmnnn}
\begin{abstract}
In this note, we study some properties of the GHZ~state.
First, we present a quantum secret sharing scheme 
in which the participants require only
classical channels in order to reconstruct the secret; our protocol is significantly  more efficient than the trivial usage of teleportation. Second, we show that the classical simulation of an $n$-party GHZ state requires at least $n\log_2 n - 2n$ bits of communication.  Finally, we present a problem simpler than the complete simulation of the multi-party GHZ state, that could lead to a no-go theorem for GHZ state simulation.
\end{abstract}
\maketitle

\section{Introduction}

The \emph{GHZ state} (also called \emph{cat} state) was  introduced by
 Daniel M. Greenberger, Michael A. Horne and Anton Zeilinger~\cite{GHZ}  as a new way of proving Bell's Theorem~\cite{Bell64}. The $n$-party version of the GHZ state is given by
\begin{equation*}
\ket{+^n} = \frac{1}{\sqrt{2}} \overbrace{\ket{00\ldots0}}^n +    \frac{1}{\sqrt{2}} \overbrace{\ket{11\ldots1}}^n  =\frac{1}{\sqrt{2}} \ket{0^n} +    \frac{1}{\sqrt{2}} \ket{1^n}  \,.
\end{equation*}

As the most frequently used multi-party entangled state, the GHZ state has appeared in applications such as nonlocality~\cite{Mermin90}, communication complexity~\cite{Tappa:1999A} and multi-party cryptography~\cite{AnonQuantum}.

Our contribution deals with the GHZ state in two scenarios. In Section~\ref{sec:secret sharing}, we show that in the context of \emph{quantum secret sharing}, the GHZ state can be used to implement an $((n,n))$-threshold scheme where the reconstruction of the secret requires only classical communication and is more efficient than the obvious protocol based quantum teleportation. In Section~\ref{sec:simulation}, we show that for the task of classical \emph{entanglement simulation}, the communication required to simulate an $n$-party GHZ state is lower-bounded by  \mbox{$n\log n -2n$}. This is an improvement on the previously known \mbox{$n\log_2 n -3n$} lower bound~\cite{BHMR03}. The general question of the feasibility of GHZ simulation is still open and, still in Section~\ref{sec:simulation}, we give a necessary condition for the task to be achievable.

\section{Secret Sharing}
\label{sec:secret sharing}

An \emph{$(n,t)$-threshold secret sharing scheme}\hspace{1mm} is a protocol by which a \emph{dealer} distributes \emph{shares} of his \emph{secret} to $n$ players such that, when combining their shares,  any subset of~$t$ or more players is able to recover the secret, while any subset of \emph{less} than~$t$ players is unable to gain any information on the secret. Classical secret sharing was independently introduced by  George Blakley~\cite{Blakley} and  Adi Shamir~\cite{Shamir}. Following the literature, we denote a \emph{quantum} threshold secret sharing scheme by~$((n,t))$~\cite{QuantumSS}, while reserving (n,t) for classical schemes.
In quantum secret sharing, it is in general assumed that, in order to reconstruct the secret,  the players have access to quantum channels. Here, we concentrate on the case where
the players \emph{do not} share quantum channels (they do however have a quantum channel with the dealer).  There is an obvious way for the  players to adapt to this restricted scenario:  quantum teleportation~\cite{teleport} enables the conversion of any  standard quantum secret sharing scheme into one with only classical communication during the reconstruction phase. This procedure substitutes each qubit of communication with two bits of communication coupled with  a pre-distributed maximally entangled two-qubit state:
\begin{equation}
\ket{\Psi^\text{-}} = \frac{1}{\sqrt{2}}\ket{01} - \frac{1}{\sqrt{2}} \ket{10}\,.
\end{equation}

\begin{theo}
\label{thm:naive}
In the teleportation-based version of a one-qubit secret sharing scheme, $\frac{n^2 -n}{2}$ shared states~$\ket{\Psi^\text{-}}$ are necessary and sufficient for the reconstruction of the secret.
\end{theo}

\begin{proof}
Because each participant is potentially the receiver of the secret, each participant must be linked to every other participant by at least one disjoint path consisting of states~$\ket{\Psi^\text{-}}$. Seeing the participants as vertices and the shared entanglement as edges, we have that each vertex  must have degree~$\geq n-1$. Counting the degree at each vertex yields a lower bound of~$n(n-1)/2$ for the total number of edges. Since the complete graph, $K_n$, satisfies our criteria, we have the desired result.
\end{proof}

If we add to Theorem~\ref{thm:naive} the requirement that each share of the secret contain a qubit, the total number of qubits required for a teleportation-based scheme is~$n^2$.
In sharp contrast, our protocol requires only a single shared multi-party state, each player holding a single qubit, for a total of~$n$ qubits. This is sufficient for \emph{both} the shares and the reconstruction. Since quantum memory is one of the most challenging aspects of experimental quantum physics, our protocol could lead to interesting implementations.
Damian Markham and Barry C.~Sanders have recently independently proposed a quantum secret sharing scheme which also uses an underlying $n$-party entangled state and only requires classical communication to reconstruct the secret~\cite{MS08}. Their approach is based on the graph state formalism.

\subsection{Protocol for Quantum Secret Sharing with Classical Reconstruction}

We now present our $((n,n))$-threshold  \emph{Quantum Secret Sharing with Classical Reconstruction} (QSS-CR) protocol. Suppose the dealer  wishes to share the quantum secret state $\ket{\Psi}=\alpha \ket{0}+\beta\ket{1}$.

\begin{enumerate}
\item\textbf{Partial encryption.} \label{step:encryption}
The dealer  chooses  uniformly at random $x \in \{0,1\}$. If $x=0$, he does nothing to~$\ket{\Psi}$ for this step.  If $x=1$, he applies the negation transformation, $N$:
  \begin{equation}
N=\left(
   \begin{array}{cc}
     0 & 1 \\
     1 & 0 \\
   \end{array}
 \right)\,. \end{equation}

  Let the resulting state be $\ket{\Psi'}= \alpha'\ket{0} + \beta' \ket{1}$\,.

\item \textbf{Expansion.} \label{step:encoding}
The dealer expands~$\ket{\Psi'}$ into an $n$-qubit state by creating $n-1$ \emph{pseudo-copies}; the resulting state~is:
 \begin{equation}
 \ket{\Psi''}  =  \alpha'\ket{0^n} + \beta' \ket{1^n}
 \end{equation}

\item \textbf{Distribution.} \label{step:distribution}
The dealer picks uniformly at random a bit string $x=x_1x_2\ldots x_n$ with $\bigoplus_{i=1}^n x_i = x$. Player~$i$'s share consists of bit~$x_i$ as well as of qubit~$i$ of~$\ket{\Psi''}$.

\item \textbf{Reconstruction} \label{step:reconstruction}
The players decide who will receive the secret; say they agree on player~1.\\

\begin{itemize}
\item  Player $i$ $(i = 2,3, \ldots, n$) applies the Hadamard transform~$H$ to his qubit:
    \begin{equation}
H= \frac{1}{\sqrt{2}}\left(
   \begin{array}{cc}
     1 & 1 \\
     1 & $-$1 \\
   \end{array}
 \right)\,. \end{equation}

\item Player $i$  $(i = 2,3, \ldots, n)$ measures his qubit in the computational basis. Let the outcome be $y_i$; this value, along with~$x_i$ is sent to player~1.

\item Player 1 computes $y = \bigoplus_{i=2}^n y_i$. If $y=0$, he does nothing. If $y=1$, he applies $Z$ to his qubit:
      \begin{equation} Z=
\left(
   \begin{array}{cc}
     1 & 0 \\
     0 & $-$1 \\
   \end{array}
 \right)\,. \end{equation}

\item
Player~1 computes $x=\bigoplus_{i=1}^n x_i$. If~$x=0$, he does nothing. If $x=1$, he applies~$N$ to his qubit. The result is the reconstructed secret.

\end{itemize}

\end{enumerate}

\subsection{Correctness and Privacy}

We now show that our QSS-CR protocol produces the correct output~(Theorem~\ref{theorem:correct}) and is secure against collusions of less than~$n$ players (Theorem~\ref{theorem:privacy}). The proof of the following theorem follows from the properties of the $GHZ$ state. \looseness=-1

\begin{theo}
\label{theorem:correct}
At the end of the QSS-CR protocol, the receiver has the initial quantum state~$\ket{\Psi}$.
\end{theo}

\begin{theo}
\label{theorem:privacy}
In the QSS-CR protocol, any subset of $s<n$ players cannot learn anything about~$\ket{\Psi}$.
\end{theo}
\begin{proof} Without loss of generality,
suppose players $1,2, \ldots ,{n-1}$ share their secrets. We now show that their joint state is independent of the initial shared secret,~$\ket{\Psi}$.
To do this, first note that the classical bits $x_1, x_2, \ldots x_{n-1}$ are uniformly distributed over all possible combinations (and independent of everything else) and in particular they reveal nothing about~$x$.
Next, note that since $\ket{\Psi''}$ is either $\alpha\ket{0^n} + \beta\ket{1^n}$ or $\beta\ket{0^n} + \alpha\ket{1^n}$ (with equal probability), the $n-1$ players can collaborate to coherently transform their joint system into a tensor product of an unknown 1-qubit state and a known $n-2$ qubit state. The unknown qubit is in the totally mixed state; it thus does not contain any information about~$\ket{\Psi}$.
\end{proof}


\section{Classical Simulation of the GHZ State}
\label{sec:simulation}

It is well known that entanglement gives rise to correlations that are not achievable by spacelike-separated parties that are allowed only prior shared randomness~\cite{Bell64}. In the study of entanglement simulation, we ask: what \emph{extra} resources are sufficient for the parties to produce correlations \emph{as if} they shared a given entangled state? 
In the case of the simulation of the maximally entangled two-qubit state~$\ket{\Psi^\text{-}}$, a single bit of communication is sufficient~\cite{TonerBacon}; the same result can also be achieved with a single use of a nonlocal box~\cite{PR,CGMP05}. In contrast to these important results, relatively little is known about the simulation of the GHZ~state, in particular it is still an open question whether or not simulation with finite communication is possible.

In Theorem~\ref{theorem:lower bound}, we give a lower bound on the number of classical bits required to simulate an $n$-party GHZ~state. Our work improves (by $n$ bits) a previous lower bound of \mbox{$n\log_2 n -3n$}~\cite{BHMR03}; our simple method is new and could provide insight into the general task of entanglement simulation.  While we still do not have an answer to the question of the existence of a simulation protocol, we now know that \emph{if} a protocol exists, it would require at least  \mbox{$n\log n_2 -2n$} bits of communication. The question of the existence of a classical simulation of the GHZ state is addressed in Section~\ref{sec:simulation}, where we give a necessary condition for a simulation to exist~(Theorem~\ref{theorem:necessary condition}).

\subsection{Lower Bound on the GHZ State Simulation}
\emph{Communication complexity} is the study of the amount of communication required in order for players to accomplish a distributed task (see, for instance~\cite{BT08}). We are interested here in the model where the complexity is counted as the number of bits that must be \emph{broadcasted} in order for every party to know the exact value of $f$ for a given input. In this section, we make links between communication complexity results and entanglement simulation. We first recall the following theorem:

\begin{theo}[\cite{Tappa:1999A}]
\label{theorem:lower bound}
There exists an
$n$-variable Boolean function~$f$ taking as inputs $k$-bit binary strings ($k > \log_2 n$)
  which,  without entanglement, has communication complexity of at least \mbox{$n\log_{2}n-n$}  bits while if the parties share prior quantum entanglement given as a GHZ~state,
 the communication complexity is $n$ bits. Furthermore, the strategy involving quantum entanglement consists of an initial round of local measurements followed by an exchange of classical messages. \\
\end{theo}

We now proceed with our main result of this section.

\begin{theo}
\label{theorem:simulation} The exact simulation of the $n$-party GHZ state requires at least \mbox{$n\log_2 n - 2n$} bits of classical communication.
\end{theo}
\begin{proof}
Let $C(n)$ be the quantity that we wish to lower bound.
Suppose it is possible to simulate a GHZ state. Then the communication complexity task of Theorem~\ref{theorem:lower bound} could be achieved by simulating the GHZ state with $C(n)$ classical bits and then communicating $n$ classical bits as in Theorem~\ref{theorem:lower bound}. Specifically:
\begin{align*}
C(n) + n &\geq n \log_2 n -n \\
C(n) &\geq n \log_2 n -2n \qedhere
\end{align*}
\end{proof}


\subsection{A Necessary Condition for GHZ State Simulation}
\label{sec:condition}

As mentioned, the possibility of GHZ state simulation with bounded communication is an open problem. Here, we give a step towards solving this problem:  a simple communication complexity task that is possible to solve \emph{if} GHZ entanglement simulation is possible. This implies that if we can somehow show that this simple task is impossible to accomplish, then the general task of GHZ simulation would also be impossible. We believe that this task somehow captures the essence of GHZ state simulation, and would be surprised if it turns our that the task is achievable, whereas the general GHZ state simulation is not.
Our new task can easily be generalized to~$n$ parties and is given by the following:

\begin{prob}
\label{problem:simulation}
Let players $P_1$, $P_2$ and $P_3$ share a random variable $\lambda$ where $0<\lambda<1$ (i.e.~the players share  unbounded random variables). A dealer gives each player an angle,  $\theta_1$, $\theta_2$ and $\theta_3$ respectively.
The goal is for the players to individually (without communication) send a message of constant length to a receiver who, after receiving all three messages, must output the value $1$ with  probability exactly \mbox{$\cos^2(\theta_1+\theta_2+\theta_3)$} and  $0$ with probability exactly  \mbox{$\sin^2(\theta_1+\theta_2+\theta_3)$}.
\end{prob}

\begin{theo}
\label{theorem:necessary condition}
The exact classical simulation of the GHZ state cannot be achieved if no protocol for Problem~\ref{problem:simulation} exists.
\end{theo}

\begin{proof}
We show the contrapositive of the statement: if an entanglement simulation protocol for the $n$-party GHZ state exists, then a protocol for  Problem~\ref{problem:simulation} exists.

Consider the following scenario: the participants initially start with a three-party GHZ state. Each party receives as input an angle $\theta_1, \theta_2$ and $\theta_3$, respectively. Each participant~$i$ applies
 \begin{equation}
P_i=\left(
   \begin{array}{cc}
     1 & 0 \\
     0 & e^{2\theta_i\sqrt{\text{-}1}} \\
   \end{array}
 \right), \end{equation}
followed by a Hadamard transform,~$H$. The resulting state just before the Hadamard transform is:
\begin{equation}
\frac{1}{\sqrt{2}} \ket{000} + \frac{e^{2 (\theta_1 + \theta_2 + \theta_3)\sqrt{\text{-}1}}}{\sqrt{2}}\ket{111}\,.
\end{equation}
Each participant measures in the computational basis and outputs the result. A simple calculation reveals that the sum of the outputs is even with probability \mbox{$\cos^2(\theta_1 + \theta_2 + \theta_3)$}, while the sum of the outputs is odd with probability $\sin^2(\theta_1 + \theta_2 + \theta_3)$.

Thus, any protocol to simulate the GHZ state must be able to simulate the above scenario. A simulation usually involves bounded classical interaction; in order to achieve the goal of Problem~\ref{problem:simulation}, all communication paths are followed simultaneously, with the receiver choosing the final correct path and computing the parity of the player's output bits.
\end{proof}

\section{Conclusion and Discussion}

We have seen how the GHZ state gives rise to an elegant and efficient quantum secret sharing protocol with purely classical communication during the reconstruction phase. Because we have significantly lowered the quantum memory requirements, our protocol may be within reach of experimental implementations. We have also shown that \emph{if} the classical simulation of the GHZ state is feasible, then it requires at least  $n\log_2 n -2n$ bits of communication. The question of whether this simulation can really be done is still open, but we have given a potential method to prove the impossibility: if we can show that Problem~\ref{problem:simulation} is impossible to achieve, then we will know that the GHZ state simulation is impossible to achieve perfectly with bounded communication. If it turns out the Problem~\ref{problem:simulation} is achievable, then we will have evidence of the possibility of GHZ state simulation.

\section{Acknowledgements}
We thank Peter H{\o}yer and Damian Markham for insightful discussions. This work was partially supported by generous funding from  \textsc{Cifar}, \textsc{Mitacs} and \textsc{Nserc}.

\newpage
\bibliographystyle{plain}
\bibliography{references}

\begin{thebibliography}{10}

\bibitem{Bell64}
J.~Bell.
\newblock On the {E}instein-{P}odolsky-{R}osen paradox.
\newblock {\em Physics}, 1:195--200, 1964.

\bibitem{teleport}
C.{\,}~H. Bennett, G.~Brassard, C.~Cr{\'e}peau, R.~Jozsa, A.~Peres, and
  W.{\,}K. Wootters.
\newblock Teleporting an unknown quantum state via dual classical and
  {E}instein-{P}odolsky-{R}osen channels.
\newblock {\em Physical Review Letters}, 70:1895--1899, 1993.

\bibitem{Blakley}
G.{\,}R. Blakley.
\newblock Safeguarding cryptographic keys.
\newblock In {\em Proceedings of the {AFIPS} National Computer Conference},
  pages 313--317, 1979.

\bibitem{AnonQuantum}
G.~Brassard, A.~Broadbent, J.~Fitzsimons, S.~Gambs, and A.~Tapp.
\newblock Anonymous quantum communication.
\newblock In {\em Proceedings of the 13th Annual International Conference on
  the Theory and Application of Cryptology \& Information Security (ASIACRYPT
  2007)}, pages 460--473, 2007.

\bibitem{BT08}
A.~Broadbent and A.~Tapp.
\newblock Can quantum mechanics help distributed computing?
\newblock In {\em {ACM} SIGACT News, Distributed Computing Column~31},
  volume~39, pages 67--76. 2008.

\bibitem{BHMR03}
H.~Buhrman, P.~H{\o}yer, S.~Massar, and H.~R{\"o}hrig.
\newblock Combinatorics and quantum nonlocality.
\newblock {\em Physical Review Letters}, 91:047903 [4 pages], 2003.

\bibitem{Tappa:1999A}
H.~Buhrman, W.~{van Dam}, P.~H{\o}yer, and A.~Tapp.
\newblock Multiparty quantum communication complexity.
\newblock {\em Physical Review~A}, 60:2737--2741, 1999.

\bibitem{CGMP05}
N.~Cerf, N.~Gisin, S.~Massar, and S.~Popescu.
\newblock Simulating maximal quantum entanglement without communication.
\newblock {\em Physical Review Letters}, 94:220403 [4~pages], 2005.

\bibitem{QuantumSS}
R.~Cleve, D.~Gottesman, and H.-K. Lo.
\newblock How to share a quantum secret.
\newblock {\em Physical Review Letters}, 83:648--651, 1999.

\bibitem{GHZ}
D.{\,}M. Greenberger, M.{\,}A. Horne, and A.~Zeilinger.
\newblock Going beyond {B}ell's theorem.
\newblock In {\em Bell's Theorem, Quantum Theory, and Conceptions of the
  Universe}, pages 69--72, 1989.

\bibitem{MS08}
D.~Markham and B.{\,}C. Sanders.
\newblock Graph states for quantum secret sharing.
\newblock Available as: \url{arXiv:0808.1532v1 [quant-ph]}, 2008.

\bibitem{Mermin90}
N.{\,}D. Mermin.
\newblock Extreme quantum entanglement in a superposition of macroscopically
  distinct states.
\newblock {\em Physical Review Letters}, 65:1838--1840, 1990.

\bibitem{PR}
S.~Popescu and D.~Rohrlich.
\newblock Quantum nonlocality as an axiom.
\newblock {\em Foundations of Physics}, 24:379--385, 1994.

\bibitem{Shamir}
A.~Shamir.
\newblock How to share a secret.
\newblock {\em Communications of the {ACM}}, 22:612--613, 1979.

\bibitem{TonerBacon}
B.{\,}F. Toner and D.~Bacon.
\newblock The communication cost of simulating {B}ell correlations.
\newblock {\em Physical Review Letters}, 91:187904 [4~pages], 2003.

\end{thebibliography}
\end{document}